\documentclass{fundam}

\publyear{25}
\papernumber{3}
\volume{194}
\issue{4}
\theDOI{10.46298/fi.14373}

\versionForARXIV

\usepackage{url} % takes care of hyperlinks, preferred over hyperref
\usepackage[ruled,lined]{algorithm2e}% provides Algorithm environment
\usepackage{graphicx}% allows for inclusion of graphic files (figures)
\usepackage{tikz}
\usepackage{amsmath}
\usepackage{amsfonts}
\usetikzlibrary{automata, positioning, arrows}

\begin{document}

\title{MNT Elliptic Curves with Non-Prime Order}

\author{Maciej Grze\'skowiak \\
Adam Mickiewicz University\\ 
Faculty of Mathematics and Computer Science\\
Uniwersytetu Poznańskiego 4, 61-614 Pozna\'n, Poland\\
maciejg{@}amu.edu.pl } 

\maketitle

\runninghead{M. Grze\'skowiak}{MNT Elliptic Curves with Non-Prime Order}

\begin{abstract}
Miyaji, Nakabayashi, and Takano    proposed an algorithm 
for the construction of prime order  pairing-friendly
elliptic curves with embedding degrees $k=3,4,6$. We present a method for generating 
generalized MNT curves. The order of such pairing-friendly curves is the product of two distinct prime numbers. 
\end{abstract}

\begin{keywords}
Pairing-based cryptography, MNT elliptic curves, Pell’s equation.
\end{keywords}

\section{Introduction}

Let $E$ be an elliptic curve defined over a finite field $\mathbb{F}_p$, where $p$ is a prime. Let $ \#E(\mathbb{F}_p)$ be the order of the group of $\mathbb{F}_p$-rational points of $E$.  Let $n \neq p$ be a prime divisor of  $\#E(\mathbb{F}_p)$.  The embedding degree of $E$ with respect to $n$ is the smallest positive integer $k$ such that $n\mid p^k-1$,
  but $n$ does not divide $p^d-1$ for $d\mid k$, $d <k$, \cite{taxonomy}. This condition is equivalent to $n>k$ divides  $\Phi_k(p)$, where $\Phi_k(x)$ is the $k$th cyclotomic polynomial.
Elliptic curves over $\mathbb{F}_p$ that have a large subgroup of prime order $n$ and a small embedding degree $k$ are commonly referred to as pairing-friendly with respect to $n$ and embedding degree $k$ \cite{taxonomy}. 

Many pairing-based cryptographic protocols require generating pairing-friendly elliptic curves. For instance: one-round three-way key exchange \cite{Joux04}, identity-based encryption \cite{BonehF03},
identity-based signature \cite{ChaC03}, and short signature schemes \cite{BonehLS01}.
From the security point of view, it is essential to find a pairing-friendly curve $E$ over $\mathbb{F}_p$  such that
the discrete logarithm problems in the group $E(\mathbb{F}_p)$, in the order $q$ subgroups of $E(\mathbb{F}_p)$, and in the multiplicative group $\mathbb{F}_{p^k}^*$ are computationally infeasible. The creators of the initial pairing-based protocols suggested utilizing supersingular elliptic curves. However, these curves are restricted to an embedding degree $k=2$ for prime fields and  $k \leq 6$ in general \cite{taxonomy}.  Therefore, ordinary curves are necessary for higher embedding degrees.

A typical pairing-friendly ordinary elliptic curve construction method consists of two main steps.
First, we find prime numbers $n$, $p$, integers $t\neq 0,1,2$ and $k \geq 3$ such that
 \begin{align}\label{warunek1}
 |t| \leq 2p^{1/2}, \quad n\mid p+1-t, \quad n\mid \Phi_k(p).
 \end{align} 
We refer to $t$ as the trace of Frobenius. By the Hasse-Weil bound $t=p+1 - \#E(\mathbb{F}_p)$ satisfies $|t|\leq2\sqrt{p}$. For every $|t |\leq 2p^{1/2}$, there exists an elliptic curve $E$ over $\mathbb{F}_p$ whose Frobenius trace is exactly $t$ (see \cite{Silverman}). In the second step, we find the equation of the curve $E$ over $\mathbb{F}_p$ with  $\#E(\mathbb{F}_p)=p+1-t$.
By (\ref{warunek1}), it is obvious that we can  write the integer
\begin{align}\label{CMeq}
t^2-4p=\Delta y^2, \quad \Delta, y \in \mathbb{Z},
\end{align}
 in the unique form, where $\Delta <0$  is a square-free integer. The above equation is called the CM equation and the integer $\Delta$ is called the CM discriminant. For given  $p, t$, the Complex Multiplication (CM) method can be used to construct the curve equation over $\mathbb{F}_p$. Unfortunately, the CM algorithm is effective if  $\Delta$  is small, that is $|\Delta| < 10^{10}$ \cite{taxonomy}. In practical applications, the number $k$ should be small, for example, $k \leq 100$, while the quotient $\log n/\log p$ should be close to one.

In \cite{mnt-curves}, Miyaji, Nakabayashi, and Takano    proposed an algorithm (the MNT method)
for the construction of prime order pairing-friendly
elliptic curves with embedding degrees $k=3,4,6$. They found families of polynomials $(n_k(x), p_k(x), t_k(x))$ in $\mathbb{Z}[x]$ satisfying
\begin{align}\label{familyMNT0}
 n_k(x) = p_k(x)+1- t_k(x), \quad n_k(x) \mid \Phi_k(p_k(x)), \quad |t_k(x)| \leq 2\sqrt{p_k(x)},
 \end{align}
(see Table \ref{MNT}). In this case, the corresponding CM equation can be written as
\begin{align*}
 t_k(x)^2 - 4p_k(x) =\Delta Y^2,\quad Y \in \mathbb{Z},
 \end{align*}
where $\Delta <0$  is a square-free integer. Multiplying the quadratic equation above by a constant factor and completing the squares, we obtain Pell’s equation
 \begin{align}\label{CMeq2}
X^2- 3\Delta Y^2 = m, \quad  m=-8, k=4,6 \quad  \mbox{or} \quad  m=24, k=3,
\end{align}
where $X=X(x)$, $Y \in \mathbb{Z}$. 
\begin{table}[h!]
\centering
\begin{tabular}{ |c|c|c|c|c|} 
 \hline
$k$ & $n_k(x)$ &$p_k(x)$& $t_k(x)$& Pell’s equation\\ \hline
6 & $4x^2\pm 2x +1$ & $4x^2+1$ & $1\pm 2x$ & $(6x \pm 1)^2+3\Delta Y^2 =-8$\\ 
4 & $x^2 + 2x +2$, $x^2+1$ & $x^2+ x +1$  & $-x$, $x+1$ &  $(3x +t)^2-3\Delta Y^2 =-8$, $t=1,2$ \\ 
3 & $12x^2\pm 6x +1$ & $12x^2 -1$ & $\pm 6x-1$ & $(6x \pm 3)^2-3\Delta Y^2 =24$\\ 
 \hline
\end{tabular}
 \caption{MNT families}
\label{MNT}
\end{table}

We will call equation (\ref{CMeq2}) generalized Pell’s equation.
This observation above leads to the MNT algorithm \cite{mnt-curves}. To find a desired curve, perform the following steps. Fix $k\in \{3,4,6\}$ and select square-free integer $|\Delta| < 10^{10}$.
Find the solution $(X_0, Y_0)$ of (\ref{CMeq2}), where $X_0=X(x_0)$, such that the corresponding numbers $n=n_k(x_0)$ and $p=p_k(x_0)$ are simultaneously primes. Finally, use the CM
method to construct the curve equation over $\mathbb{F}_p$. For a deeper discussion of the theory of Pell equations,
we refer the reader to \cite{Mollin}.

Luca and Shparlinski \cite{LucaSpar} gave some heuristic estimates on the number of elliptic curves which can be
produced by the MNT algorithm. Let $E(z)$ denote  the expected total number of all
isogeny classes of MNT curves over all finite fields with
embedding degree $k$ and CM discriminant $|\Delta| \leq z$. Then we have
\begin{align*}
E(z) \ll \frac{z}{(\log z)^2}.
\end{align*}
From the above estimate, the elliptic curves generated by the MNT algorithm are rare. We refer the reader to \cite{LucaSpar} for a deeper discussion of the lower bound of the generalized version of the function $E(z)$.

On the other hand, in most applications, an elliptic curve with $ \#E(\mathbb{F}_p)=qn$ is acceptable, where $q$ is small.
Barreto and Scott used this idea in \cite{BaretoScott}. In particular, they extended the MNT algorithm to construct more Pell equations for $q > 1$. Galbraith, McKee, and Valença \cite{GALBRAITH2007800} generalize the MNT method by giving
families of ordinary curves corresponding to non-prime group orders $\#E(\mathbb{F}_p)=qn$ with a prime $n$, $q=2,3,4,5$ and $k=3,4,6$.
Fotiadis and Konstantinou \cite{Konnstatinou} extend the search to the MNT ordinary families with larger no prime cofactors $5<q <48$, and $k=3,4,6$. In  \cite{Duan}, the authors propose a general algorithm for constructing pairing-friendly elliptic curves with an arbitrary embedding degree. For a treatment of a more general case construction of pairing-friendly curves, we refer the reader to \cite{taxonomy}. Let us fix a positive integers $h$ and $k \in \{3,4,6\}$. Let $\mathcal{F}_k(h)$ denote the set of  all possible families of  MNT curves  $(n_k(x), p_k(x), t_k(x))$ corresponding to group orders $\#E(\mathbb{F}_p)=qn$ with embedding degrees $k$, where  $1\leq q \leq h$. In \cite{OnNearMNT}, an algorithm was presented that outputs $\mathcal{F}_k(h)$. The algorithm idea is based on the following observations.
From (\ref{familyMNT0}) we have,
\begin{align*}
\Phi_k(p_k(x)) \equiv \Phi_k(t_k(x)-1) \mod{n_k(x)}.
\end{align*}
Choose  $a, b \in \mathbb{Z}$. If $t_k(x)=ax+b$,  then (\ref{familyMNT0}) shows that the degree of the polynomial $n_k(x)$ is equal to 2.  Therefore, we can write,
\begin{align*}
\Phi_k(t_k(x)-1)= d \cdot n_k(x), \quad d \in \mathbb{Z},
\end{align*}
where $d$ is the smallest common multiple of the coefficients of $\Phi_k(t_k(x)-1)$.
Now, assume that $n_k(x)$ is irreducible over $\mathbb{Z}$ and there exists  $q \in \mathbb{Z}$ such that
\begin{align}\label{algoritmA}
p_k(x)= q\cdot n_k(x) + t_k(x)-1,
\end{align}
where $p_k(x)$ is irreducible over $\mathbb{Z}$. If moreover,   
\begin{align}\label{algoritmB}
|t_k(x)| \leq 2\sqrt{p_k(x)},
\end{align}
then we can use the polynomials $(n_k(x), p_k(x), t_k(x))$ above to construct the corresponding an elliptic MNT curve with $q$ dividing  $ \#E(\mathbb{F}_p)$.
From the above observations, it is easy to construct an algorithm that finds all possible pairs of integer numbers $a$ and $b$ such that the corresponding conditions (\ref{algoritmA}) and  (\ref{algoritmB}) are satisfied (see \cite{OnNearMNT}). Given embedding degree $k$ and a positive integer $h$, the algorithm tests all possible pairs of integers $|a| \leq 4h $ and $|b| < 4h$ to determine whether they satisfy the following condition.
For fixed $a, b$, is there a positive integer $1\leq q < h$ satisfying (\ref{algoritmA}),  (\ref{algoritmB}) such that the corresponding polynomials $n_k(x)$ and $p_k(x)$ are irreducible over $\mathbb{Z}$?
If the condition is satisfied, for such $a$ and $b$ a family of polynomials $(n_k(x), p_k(x), t_k(x))$ is computed.
We refer the reader to \cite{OnNearMNT} for more details.
It is easy to see that the algorithm takes no more than $O(h^3)$ steps to check all possible integers $|a| \leq 4h $, $|b| < 4h$ with $1\leq q < h$. Therefore, given $h$, calculating $\mathcal{F}_k(h)$ 
%a list of all possible families  of ordinary elliptic curves corresponding to group orders $\#E(\mathbb{F}_p)=qn$ with $k=3,4,6$
using the algorithm shown in \cite{OnNearMNT} may require an exponential number of steps. In \cite{OnNearMNT}, the authors computed the corresponding families of polynomials for $h =6$ and $k=3,4,6$. Now, we introduce the following definition.

\begin{definition}\label{DefGenMNT}
Fix $k \in \{3,4,6\}$ and  a prime $q \equiv 1 \pmod k$. 
The triple  $(n_k(x), p_k(x), t_k(x))$ polynomials in  $\mathbb{Z}[x]$ parameterizes a family of generalized MNT elliptic curves with embedding degree $k$ if 
\begin{align}\label{familyMNT}
 q n_k(x) = p_k(x)+1- t_k(x), \quad qn_k(x) \mid \Phi_k(p_k(x)), \quad t_k(x)^2-4p_k(x)\leq0 ,
\end{align}
and polynomials $n_k(x), p_k(x)$ are irreducible over $\mathbb{Z}$.
\end{definition}

\begin{remark}
We see at once that if there is $x_0 \in \mathbb{Z}$  such that $n=n_k(x_0)$, and $p=p_k(x_0)$ are simultaneously prime,
 then there exists elliptic curve $E$ defined over finite field $\mathbb{F}_{p}$ such that
\begin{align*}
 \#E(\mathbb{F}_{p})=qn=p+1- t, \quad t= t_k(x_0).
\end{align*}
\end{remark}

\begin{remark}
Let $\#E(\mathbb{F}_{p})=qn$. If $E(\mathbb{F}_{p})$ has an embedding degree $k$ with respect to $n$, then an embedding degree with respect to $q$ can be generally different from $k$. Definition \ref{DefGenMNT} covers the case where the embedding degree of $E(\mathbb{F}_{p})$ equals $k$ for any prime divisor of $\#E(\mathbb{F}_{p})$ (see Proposition 2.4, \cite{taxonomy}). It is clear that  the families of polynomials satisfying Definition \ref{DefGenMNT} belong to $\mathcal{F}_k(h)$. 
\end{remark}

The present paper extends the idea of effective polynomial families, first introduced in  \cite{mnt-curves}.
Our  method generates families of polynomials that satisfy the properties of Definition \ref{DefGenMNT}.
In particular, we propose methods for generating families of ordinary curves corresponding to non-prime group orders when $q$  is any given prime number. By including an infinite family of prime cofactors in the analysis, we obtain a class of polynomial families belonging to $\mathcal{F}_k(h)$. For a given $k \in \{3,4,6\}$ and prime number $q <h$, our method finds a single family of polynomials of $\mathcal{F}_k(h)$ in a polynomial time with respect to the number of bits $h$. It is enough to calculate the root $x \pmod q$ of a given explicit quadratic polynomial to do so (see Theorems below). This property significantly speeds up the algorithm presented in the paper \cite{OnNearMNT}, which systematically calculates all possible solutions to the problem and checks each for a valid solution. For given a list of large prime numbers $q_i < h$, with $i = O(\log^c h), c>0$ and $k$, the approach presented in this paper allows us to efficiently determine the corresponding class of families of polynomials in $\mathcal{F}_k(h)$ while using the algorithm from  \cite{OnNearMNT} to this task would require exponential time with respect to $h$. We provide the corresponding generalized Pell equation for the constructed families to construct desired elliptic curves effectively. All this together allows us to build an algorithmic method analogous to the algorithm in  \cite{mnt-curves}.

 The remaining part of the paper is organized as follows. In Section \ref{MainT}, our families of polynomials are presented.  Section \ref{ProofT} contains a detailed analysis of our constructions.

\section{Main theorems}\label{MainT}
Throughout this paper,  $\Delta < 0$ is a square-free rational integer. We denote by $\mathbb{Z}$ the ring of integers numbers. Let $k$ be a positive integer, and let $\Phi_k(x) \in \mathbb{Z}[x]$ be the $k$th cyclotomic polynomial; this is a unique monic polynomial  of degree $\varphi(k)$ whose roots are the complex primitive $k$th roots of unity, where $\varphi$ is Euler's totient function. In this article, we will consider only the case $k=,3,4,6$. For the convenience of the reader, we recall that 
\begin{align*}
\Phi_3(x)=x^2+x+1, \quad \Phi_4(x)=x^2+1, \quad \Phi_6(x)=x^2-x+1.
\end{align*}
In the following subsections, we will present parametric families of polynomials that are useful in constructing generalized MNT elliptic curves over a finite field with an embedding degree $k$.

\subsection{The case $k=6$}
\begin{theorem}\label{th1}
Fix $j\in \{3,6\}$, a prime $q \equiv 1 \pmod 6$ or $q=3$. Let  $s <q$ be a root of   $\Phi_j(x)\pmod q$.  
If $p_6(x) = \Phi_4(qx+s)$, 
\begin{align*}
n_6(x)= \left\{
\begin{array}{ccc}
qx^2+(2s+1)x+\Phi_3(s)/q, & \quad  \mbox{if}  & \quad q \mid \Phi_3(s),  \\
qx^2+(2s-1)x+\Phi_6(s)/q , & \quad  \mbox{if}  & \quad q \mid \Phi_6(s).
\end{array}
\right .
\end{align*}
and
\begin{align*}
t_6(x)= \left\{
\begin{array}{ccc}
1-qx-s, & \quad  \mbox{if}  & \quad q \mid \Phi_3(s),  \\
1+ qx+s, & \quad  \mbox{if}  & \quad q \mid \Phi_6(s),
\end{array}
\right .
\end{align*}
then  polynomials $(n_6(x),p_6(x), t_6(x))$ parameterizes a family of generalized MNT elliptic curves with embedding degree 6. Moreover, the family has the corresponding generalized Pell equations 
\begin{align*}
X^2+3\Delta Y^2=-8, \quad X=
 \left\{
\begin{array}{ccc}
3(qx+s)+1 & \quad  \mbox{if}  & \quad q \mid \Phi_3(s),  \\
3(qx+s) -1 & \quad  \mbox{if}  & \quad q \mid \Phi_6(s).
\end{array}
\right.
\end{align*}
\end{theorem}
\begin{proof}
See Section \ref{Section6}.
\end{proof}
\begin{remark}
Taking $j\in \{3, 6\}$, $q=1$, $s=0$, and $x=\pm 2y$ in  Theorem \ref{th1}, we obtain the MNT family with embedding degree 6.
\end{remark}
\begin{remark}
Taking $q=3$, $s=1$  in Theorem \ref{th1}, we get family belonging to the set $\mathcal{F}_6(6)$
(see Table 4,  \cite{OnNearMNT}).
\end{remark}

\subsection{The case  $k=4$}

\begin{theorem}\label{th2}
Fix a prime $q \equiv 1 \pmod 4$ or $q=2$. Let  $s <q$ or $s -1 <q$ be a root of   $\Phi_4(x) \pmod q$.  
If $p_4(x) = \Phi_6(qx+s)$, 
\begin{align*}
n_4(x)= \left\{
\begin{array}{lll}
qx^2+2sx+\Phi_4(s)/q, & \quad  \mbox{if}  & \quad q \mid \Phi_4(s),  \\
qx^2+(2s-2)x+\Phi_4(s-1)/q, & \quad  \mbox{if}  & \quad q \mid \Phi_4(s -1).
\end{array}
\right .
\end{align*}
and
\begin{align*}
t_4(x)= \left\{
\begin{array}{lll}
1-qx-s, & \quad  \mbox{if}  & \quad q \mid \Phi_4(s),  \\
qx+s, & \quad  \mbox{if}  & \quad q \mid \Phi_4(s -1),
\end{array}
\right .
\end{align*}
then  polynomials $(n_4(x),p_4(x), t_4(x))$ parameterizes a family of generalized MNT elliptic curves with embedding degree 4. Moreover, the family has the corresponding generalized Pell equations 
\begin{align*}
X^2+3\Delta Y^2=-8, \quad X=
 \left\{
\begin{array}{lll}
3(qx+s)+1 & \quad  \mbox{if}  & \quad q \mid \Phi_4(s),  \\
3(qx+s)+2 & \quad  \mbox{if}  & \quad q \mid \Phi_4(s -1).
\end{array}
\right.
\end{align*}
\end{theorem}
\begin{proof}
See Section \ref{Section4}.
\end{proof}
\begin{remark}
Taking $q=1$, $s=0$, and $x=\pm y$ in  Theorem \ref{th2}, we obtain the MNT family with embedding degree 4.
\end{remark}

\begin{table}[h!]
\centering
\begin{tabular}{ |c|c|c|c|c|} 
 \hline
$q$ & $ s$ &$n_k(x)$ &$p_k(x)$& $t_k(x)$\\ \hline
2 & 1 &$2x^2+ 2x +1$ & $4x^2+2x+1$ & $ -2x$ \\ 
5 & 2 & $5x^2 + 4x +1$, & $25x^2+15x+3$  & $-5x-1$  \\ 
5 & 3  &$5x^2 + 6x +2$ & $25x^2+25x+7$ & $-5x-1$ \\ 
 \hline
\end{tabular}
 \caption{MNT families}
\label{MNTG}
\end{table}
\begin{remark}
Taking $q$ and $s$ as in Table 
\ref{MNTG} and applying them to Theorem \ref{th2}, we immediately get families belonging to the set $\mathcal{F}_4(6)$
(see Table 3,  \cite{OnNearMNT}).
\end{remark}

\subsection{The case $k=3$}

\begin{theorem}\label{th3}
Let $g_0(x)=3x^2-1$, $g_1(x) = 3x^2-3x+1, g_2(x) = 3x^2+3x+1  \in \mathbb{Z}[x]$. Fix  a prime $q \equiv 1 \pmod 3$, and let
$s <q$ be a root of   $g_1(x) \pmod q$ or $g_2(x) \pmod q$.
 If $p_3(x) = g_0(qx+s)$, 
\begin{align*}
n_3(x)= \left\{
\begin{array}{lll}
3qx^2+(6s-3)x+g_1(s)/q, & \quad  \mbox{if}  & \quad q \mid g_1(s),  \\
3qx^2+(6s+3)x+g_2(s)/q, & \quad  \mbox{if}  & \quad q \mid g_2(s),
\end{array}
\right .
\end{align*}
and
\begin{align*}
t_3(x)= \left\{
\begin{array}{lll}
3(qx+s)-1, & \quad  \mbox{if}  & \quad  q \mid g_1(s),  \\
1-3(qx+s), & \quad  \mbox{if}  & \quad q \mid g_2(s),
\end{array}
\right .
\end{align*}
then  polynomials $(n_3(x),p_3(x), t_3(x))$ parameterizes a family of generalized MNT elliptic curves with embedding degree 3. Moreover, the family has the corresponding generalized Pell equations 
\begin{align*}
X^2+3\Delta Y^2=24, \quad X=3(qx+s)+3.
\end{align*}
\end{theorem}
\begin{remark}
Taking $q=1$, $s=0$, and $x=\pm 2y$ in  Theorem \ref{th3}, we obtain the MNT family with embedding degree 3.
\end{remark}

\section{Proof of Theorems}\label{ProofT}
\subsection{The case  $k=6$}\label{Section6}
\begin{lemma}\label{phi6phi4}
Fix  a prime $q \equiv 1 \pmod 6$ or $q=3$.  Let  $\Phi_6(s) \equiv 0 \pmod q$ or $\Phi_3(s) \equiv 0 \pmod q$. Then we have,
\begin{align}\label{podzielnosc6}
\left\{
\begin{array}{ccc}
\Phi_6( \Phi_4(qx+s))  =   qf_1(x)f_2(x), & \quad  \mbox{if}  & \quad q \mid \Phi_3(s),  \\
\Phi_6( \Phi_4(qx+s))  =   qf_3(x)f_4(x), & \quad  \mbox{if}  & \quad q \mid \Phi_6(s),  
\end{array}
\right .
\end{align}
for  $x \in \mathbb{Z} $ and polynomials $f_i(x)$ are  irreducible over  $\mathbb{Z}$, $i=1,2,3,4$, where
\begin{align*}
f_1(x) & =qx^2+(2s+1)x+\Phi_3(s)/q, \quad f_2(x)  =q^2x^2+(2s-1)qx+\Phi_6(s),\\
f_3(x) & =q^2x^2+(2s+1)qx+\Phi_3(s), \quad f_4(x)  =qx^2+(2s-1)x+\Phi_6(s)/q.
\end{align*}
\end{lemma}
\begin{proof}
If $q \equiv 1 \pmod 6$, then $-3$ is a quadratic residue $\pmod q$, and a root of $\Phi_j(x) \pmod q$ can be computed, $j=3,6$. It is easily seen that $\Phi_3(1)\equiv \Phi_6(2)\equiv 0 \pmod 3$. A trivial verification shows that,
\begin{align}\label{rownanie64}
\Phi_6(\Phi_4(x))=\Phi_3(x)\Phi_6(x), \quad x \in \mathbb{Z}.
\end{align}
Let $s$ be a root of $\Phi_k(x) \pmod q$, $k=3$ or $k=6$. From (\ref{rownanie64}) it follows that,
\begin{align*}
\Phi_6( \Phi_4(qx+s)) = \Phi_3(qx+s)\Phi_6(qx+s) = \left\{
\begin{array}{ccc}
qf_1(x)f_2(x), & \quad  \mbox{if}  & \quad q \mid \Phi_3(s),  \\
qf_3(x)f_4(x), & \quad  \mbox{if}  & \quad q \mid \Phi_6(s), 
\end{array}
\right .
\end{align*}
where
\begin{align*}
f_1(x) & =qx^2+(2s+1)x+\Phi_3(s)/q, \quad f_2(x)  =q^2x^2+(2s-1)qx+\Phi_6(s),\\
f_3(x) & =q^2x^2+(2s+1)qx+\Phi_3(s), \quad f_4(x)  =qx^2+(2s-1)x+\Phi_6(s)/q.
\end{align*}
The polynomials $f_i$ are irreducible over $\mathbb{Z}$, $i=1,2,3,4$. Indeed, $\Delta(f_i)$ the discriminants of $f_i$  are negative,
\begin{align*}
\Delta(f_1) & =(2s+1)^2-4\Phi_3(s)=-3, \quad \Delta(f_2)   =q^2((2s-1)^2-4\Phi_6(s))=-3q^2,\\
\Delta(f_3) & =q^2((2s+1)^2-4\Phi_3(s))=-3q^2, \quad \Delta(f_4)   =(2s-1)^2-4\Phi_6(s)=-3.
\end{align*}
This finishes the proof.
\end{proof}
We are now in a position to prove Theorem \ref{th1}.
\begin{proof}
Let $q \equiv 1 \pmod 6$ be a prime or $q=3$, and let $\Phi_3(s)\equiv 0 \pmod q$. We will show that polynomials $n_6(x), p_6(x)$ and $t_6(x)$ satisfy the conditions (\ref{familyMNT}).
 We have,
\begin{align*}
qn_6(x)& =q^2x^2+(2s+1)qx+\Phi_3(s)=\Phi_4(qx+s)+qx+s\\
& = p_6(x)+ 1 - t_6(x),
\end{align*}
so $qn_6(x) \mid p_6(x)+1-t_6(x)$. An easy computation shows that,
\begin{align}\label{CMmoment}
t_6(x)^2-4p_6(x) & =-3(qx+s)^2-2(qx+s)-3<0.
\end{align}
Since $n_6(x)=f_1(x)$, (\ref{podzielnosc6}) shows that
\begin{align*}
qn_6(x) \mid \Phi_6(p_6(x)).
\end{align*}
The polynomials $n_6(x)$ and $p_6(x)$ are irreducible over $\mathbb{Z}$, which is clear from Lemma~\ref{phi6phi4} and  is easy to check. So the polynomials  $n_6(x), p_6(x)$ and  $t_6(x)$ satisfy Definition \ref{DefGenMNT}.
Fix $x$ for the moment. We can write (\ref{CMmoment}) in the form
\begin{align*}
t_6(x)^2-4p_6(x) & =-3(qx+s)^2-2(qx+s)-3 = \Delta Y^2, \quad Y \in \mathbb{Z},
\end{align*}
where $\Delta <0$  is a square-free integer. Multiplying the above equation by -3, we obtain
\begin{align*}
X^2+3\Delta Y^2=-8, \quad X=3(qx+s)+1.
\end{align*}
The same proof works if $\Phi_6(s)\equiv 0 \pmod q$.  The details are left to the reader. This finishes the proof.
\end{proof}

\subsection{The case  $k=4$}\label{Section4}
\begin{lemma}\label{phi4phi6}
Fix  a prime $q \equiv 1 \pmod 4$ or $q=2$. If $\Phi_4(s) \equiv 0 \pmod q$ or $\Phi_4(s-1) \equiv 0 \pmod q$. Then we have
\begin{align*}
\left\{
\begin{array}{ccl}
\Phi_4( \Phi_6(qx+s))  =   qf_5(x)f_6(x), & \quad  \mbox{if}  & \quad q \mid \Phi_4(s),  \\
\Phi_4( \Phi_6(qx+s))  =   qf_7(x)f_8(x), & \quad  \mbox{if}  & \quad q \mid \Phi_4(s-1), 
\end{array}
\right .
\end{align*}
for  $x \in \mathbb{Z} $, and polynomials $f_i(x) \in \mathbb{Z}[x]$ are  irreducible over  $\mathbb{Z}$, $i=5,6,7,8$, where
\begin{align*}
f_5(x) & =qx^2+2sx+\Phi_4(s)/q, \quad f_6(x)  =q^2x^2+(2s-2)qx+\Phi_4(s-1),\\
f_7(x) & =q^2x^2+2sqx+\Phi_4(s), \quad f_8(x)  =qx^2+(2s-2)x+\Phi_4(s-1)/q.
\end{align*}
\end{lemma}

\begin{proof}
If $q \equiv 1 \pmod 4$, then $-1$ is a quadratic residue $\pmod q$, and a root of $\Phi_4(x) \pmod q$ can be computed. It is easily seen that $\Phi_2(1)\equiv 0 \pmod 3$. A trivial verification shows that,
\begin{align}\label{rownanie46}
\Phi_4(\Phi_6(x))=\Phi_4(x)\Phi_4(x-1), \quad x \in \mathbb{Z}.
\end{align}
Let $s$ be a root of  $\Phi_4(x) \pmod q$ or let $\Phi_4(s-1)\equiv 0 \pmod q$. From (\ref{rownanie46}) it follows that,
\begin{align}\label{podzielnosc4}
\Phi_4( \Phi_6(qx+s)) = \left\{
\begin{array}{ccc}
qf_5(x)f_6(x), & \quad  \mbox{if}  & \quad q \mid \Phi_4(s),  \\
qf_7(x)f_8(x), & \quad  \mbox{if}  & \quad q \mid \Phi_4(s-1), 
\end{array}
\right .
\end{align}
where
\begin{align*}
f_5(x) & =qx^2+2sx+\Phi_4(s)/q, \quad f_6(x)  =q^2x^2+(2s-2)qx+\Phi_4(s-1),\\
f_7(x) & =q^2x^2+2sqx+\Phi_4(s), \quad f_8(x)  =qx^2+(2s-2)x+\Phi_4(s-1)/q.
\end{align*}
The polynomials $f_i$ are irreducible over $\mathbb{Z}$, $i=5,6,7,8$. Indeed, $\Delta(f_i)$ the discriminants of $f_i$  are negative,
\begin{align*}
\Delta(f_5) & =4s^2-4\Phi_4(s)=-4, \quad \Delta(f_6)   =q^2((2s-2)^2-4\Phi_4(s-1))=-4q^2,\\
\Delta(f_7) & =q^2(4s^2-4\Phi_4(s))=-4q^2, \quad \Delta(f_8)   =(2s-2)^2-4\Phi_4(s-1)=-4.
\end{align*}
This finishes the proof.
\end{proof}
We are now in a position to prove Theorem \ref{th2}.
\begin{proof}
Fix  a prime $q \equiv 1 \pmod 4$ or $q=2$, and let $\Phi_4(s) \equiv 0 \pmod q$.
We will show that polynomials $n_4(x), p_4(x)$ and $t_4(x)$ satisfy the conditions (\ref{familyMNT}).
We have,
\begin{align*}
qn_4(x)& =q^2x^2+2sqx+\Phi_4(s)=(qx+s)^2+1\\
& = \Phi_6(qx+s)+(qx+s)=p_4(x)+ 1 - t_4(x)
\end{align*}
so $qn_4(x) \mid p_4(x)+1-t_4(x)$. A trivial verification shows that
\begin{align}\label{CMmoment4}
t_4(x)^2-4p_4(x) & =-3(qx+s)^2-2(qx+s)-3<0.
\end{align}
Since $n_4(x)=f_5(x)$, (\ref{podzielnosc4}) shows that
\begin{align*}
qn_4(x) \mid \Phi_4(p_4(x)).
\end{align*}
The polynomials $n_4(x)$ and $p_4(x)$ are irreducible over $\mathbb{Z}$, which is clear from Lemma~\ref{phi4phi6} and  is easy to check. So the polynomials  $n_4(x), p_4(x)$ and  $t_4(x)$ satisfy Definition \ref{DefGenMNT}.
Fix $x$ for the moment. We can write (\ref{CMmoment4}) in the form
\begin{align*}
t_4(x)^2-4p_4(x) & =-3(qx+s)^2-2(qx+s)-3=\Delta Y^3, \quad Y \in \mathbb{Z},
\end{align*}
where $\Delta <0$  is a square-free integer. Multiplying the above equation by -3, we obtain
\begin{align*}
X^2+3\Delta Y^2=-8,\quad X=3(qx+s)+1
\end{align*}
The same proof works if $\Phi_4(s-1) \equiv 0 \pmod q$.  The details are left to the reader. This finishes the proof.
\end{proof}

\subsection{The case  $k=3$}

\begin{lemma}\label{phi3phi?}
 Let $g_0(x)=3x^2-1$, $g_1(x) = 3x^2-3x+1$ and $g_2(x) = 3x^2+3x+1 \in \mathbb{Z}[x]$. Fix  a prime $q \equiv 1 \pmod 6$, and let $g_1(s) \equiv 0 \pmod q$ or $g_2(s) \equiv 0 \pmod q$. Then we have,
\begin{align}\label{podzielnosc3}
\left\{
\begin{array}{ccl}
\Phi_3(g_0(qx+s))  =   qf_9(x)f_{10}(x), & \quad  \mbox{if}  & \quad q \mid g_1(s),  \\
\Phi_3( g_0(qx+s))  =   qf_{11}(x)f_{12}(x), & \quad  \mbox{if}  & \quad q \mid g_2(s)  
\end{array}
\right .
\end{align}
$x \in \mathbb{Z}$ and polynomials $f_i(x) \in \mathbb{Z}[x]$ are  irreducible over  $\mathbb{Z}$, $i=9,10,11, 12$, where
\begin{align*}
f_9(x) & =3qx^2+(6s-3)x+g_1(s)/q, \quad f_{10}(x)  =3q^2x^2+(6s+3)qx+g_2(s),\\
f_{11}(x) & =3q^2x^2+(6s-3)qx+g_1(s), \quad f_{12}(x)  =3qx^2+(6s+3)x+g_2(s)/q.
\end{align*}
\end{lemma}

\begin{proof}
If $q \equiv 1 \pmod 6$, then $-3$ is a quadratic residue $\pmod q$, and a root of $g_j(x) \pmod q$ can be computed, $j=1,2$. A trivial verification shows that,
\begin{align}\label{rownanie3?}
\Phi_3( g_0(x)) = g_1(x)g_2(x), \quad x \in \mathbb{Z}. 
\end{align}
Let $s$ be a root of $g_k(x) \pmod q$, $j=1$ or $j=2$. From (\ref{rownanie3?}) it follows that,
\begin{align*}
\left\{
\begin{array}{ccl}
\Phi_3(g_0(qx+s))  =   qf_9(x)f_{10}(x), & \quad  \mbox{if}  & \quad q \mid g_1(s),  \\
\Phi_3( g_0(qx+s))  =   qf_{11}(x)f_{12}(x), & \quad  \mbox{if}  & \quad q \mid g_2(s),  
\end{array}
\right .
\end{align*}
where
\begin{align*}
f_9(x) & =3qx^2+(6s-3)x+g_1(s)/q, \quad f_{10}(x)  =3q^2x^2+(6s+3)qx+g_2(s),\\
f_{11}(x) & =3q^2x^2+(6s-3)qx+g_1(s), \quad f_{12}(x)  =3qx^2+(6s+3)x+g_2(s)/q.
\end{align*}
The polynomials $f_i(x)$ are  irreducible over  $\mathbb{Z}$, $i=9,10,11, 12$. Indeed, $\Delta(f_i)$ the discriminants of $f_i$  are negative,
\begin{align*}
\Delta(f_9) & =(6s-3)^2-12g_1(s)=-3, \quad \Delta(f_{10})   =q^2((6s+3)^2-12g_2(s))=-3q^2,\\
\Delta(f_{11}) & =q^2((6s-3)^2-12q_1(s))=-3q^2, \quad \Delta(f_{12})   =(6s+3)^2-12g_2(s)=-3.
\end{align*}
This finishes the proof.
\end{proof}
We are now in a position to prove Theorem \ref{th3}.
\begin{proof}
Let $q \equiv 1 \pmod 6$ be a prime, and let $g_1(s)\equiv 0 \pmod q$. We will show that polynomials $n_3(x), p_3(x)$ and $t_3(x)$ satisfy the conditions (\ref{familyMNT}).
 We have,
\begin{align*}
qn_3(x)& =3q^2x^2+(6s-3)qx+g_1(s)=p_3(qx+s)-3(qx+s)-\\
& = p_3(x)+ 1 - t_3(x)
\end{align*}
so $qn_3(x) \mid p_3(x)+1-t_3(x)$. An easy computation shows that,
\begin{align}\label{CMmoment3}
t_3(x)^2-4p_3(x) & =-3(qx+s)^2-6(qx+s)+5<0.
\end{align}
Since $n_3(x)=f_9(x)$, (\ref{podzielnosc3}) shows that
\begin{align*}
qn_3(x) \mid \Phi_3(p_3(x)).
\end{align*}
The polynomials $n_3(x)$ and $p_3(x)$ are irreducible over $\mathbb{Z}$, which is clear from Lemma~\ref{phi3phi?} and  is easy to check. So the polynomials  $n_3(x), p_3(x)$ and  $t_3(x)$ satisfy Definition \ref{DefGenMNT}.
Fix $x$ for the moment. We can write (\ref{CMmoment3}) in the form
\begin{align*}
t_3(x)^2-4p_3(x) & =-3(qx+s)^2-6(qx+s)+5 = \Delta Y^2, \quad Y \in \mathbb{Z},
\end{align*}
where $\Delta <0$  is a square-free integer. Multiplying the above equation by -3, we obtain
\begin{align*}
X^2+3\Delta Y^2=24, \quad X=3(qx+s)+3.
\end{align*}
The same proof works for  $g_2(s)\equiv 0 \pmod q$.  The details are left to the reader. This finishes the proof.
\end{proof}

%\section{An illustrative example}\label{Example}
\bibliographystyle{fundam}
\bibliography{citations}

%%%%%%%%%%%%%%%%%%%%%%%%%%%%%%%%%%%%%%%%%%%%%%%%%%%%%%%%%%%%%%%%%%%%%%
\end{document}